\documentclass[aps,pra,reprint,twoside]{revtex4-2}
\usepackage[english]{babel}

\usepackage{graphicx}
\usepackage{dcolumn}
\usepackage{bm}
\usepackage{physics}
\usepackage{amsmath, amssymb, amsfonts, amsthm, dsfont}
\usepackage{svg}
\usepackage{color}
\usepackage{url}
\usepackage[
  colorlinks=true,
  linkcolor=blue,   
  citecolor=blue,   
  urlcolor=blue     
]{hyperref}
\usepackage[capitalise]{cleveref}
\usepackage{microtype}
\usepackage[normalem]{ulem}
\usepackage{stmaryrd}
\usepackage{multirow}
\usepackage{booktabs}
\usepackage{thmtools}

\definecolor{darkblue}{RGB}{50,10,180}
\definecolor{orangef}{RGB}{210,100,20}

\newtheorem{Definition}{Definition}

\newtheorem{Proposition}{Proposition}

\begin{document}

\title{Constant sized support state distillation with qubit recycling}

\author{Victor Barizien}
\affiliation{Department of Applied Physics, University of Geneva, Geneva, Switzerland}

\date{\today}

\begin{abstract}
In order to circumvent that no universal set of gates can be achieved on a given quantum error-correcting codes, additional gates can be performed through magic state injection. This however requires to have access to high fidelity magic state, which can be obtained by the distillation of low quality states to fewer higher quality ones. In this paper, we provide a family of distillation protocols able to produce diagonal states at every level of the Clifford hierarchy. This family is obtained recursively using code doubling techniques and leveraging qubit recycling, the fact that some idle qubits of the protocols can be measured and re-used. At a given level of the Clifford hierarchy, we show that the protocols we derive can be performed at arbitrary high distance on a fix number of logical qubits. This family recovers many known efficient protocols, and uncovers new ones, such as a $111 \to 1$ protocol for $\ket{T}$ state distillation at distance $7$, which are the most compact in term of volume. Finally, we discuss possible extension of this framework for distillation protocols producing more than one output state.
\end{abstract}

\maketitle

\section{Introduction}
Quantum error correction is a central ingredient for realizing fault-tolerant quantum computation~\cite{Campbell_2017}. However, error correction alone does not provide a universal set of fault-tolerant logical gates. In particular, the Eastin--Knill theorem~\cite{EastinKnill2009} rules out a universal set of transversal gates for any quantum error-correcting code correcting all local errors. More generally, a number of results constrain the logical gates that can be implemented transversally on a given code~\cite{Zeng08,Pastawski15}. As such, most standard architectures are based on codes that cannot efficiently perform non-Clifford operations, as for the surface code~\cite{Dennis2002, Litinski19b, Fowler19} and the bicycle bivariate code~\cite{Bravyi24}. 

A standard approach to performing these additional operations is through magic state injection~\cite{Bravyi05}, performing the desired gate using an ancillary state. To avoid introducing errors, these magic states need to be of high quality so that no error propagates. Useful quantum computations typically require more than $10^8$ non-Clifford gates, so injected magic states needs to have a error rate below $10^-10$ to keep the computation reliable~\cite{Beverland22}. However, current magic state preparations can only achieve error rate of $10^-3$~\cite{Li15, Lao22}. State distillation provides protocols to circumvent this: consuming several noisy copies of a magic state and, conditioned on the detection of potential errors, producing fewer copies with substantially reduced error probability~\cite{Bravyi05, Campbell2012,Jones2013,Campbell17}. A key parameter of such protocol is its distance $d$, encoding that the output error scales as $p_{\mathrm{out}}=O(p^d)$, where $p$ is the error probability of the input states. In turn, magic state distillation is the bottleneck of fault-tolerant computation and reducing the amount of additional resources of such protocol is thus an important problem~\cite{BravyiHaah2012,Haah17b,Litinski2019}.

State distillation protocols were initially formulated by leveraging another quantum error-correcting codes admitting a transversal implementation of the target gate~\cite{Bravyi05}. For the $T=\sqrt{\sqrt{Z}}$ gate, a prominent example of non-Clifford operation, this led in particular to the study of triorthogonal codes~\cite{BravyiHaah2012}. However, recent works have shown that such distillation protocols can directly be obtained from a binary matrix satisfying appropriate orthogonality conditions~\cite{Litinski2019}. A framework proposed in Ref.~\cite{Jacinto26} maps such $N \times n$ binary matrices to physical implementations of the protocol as quantum circuits on $N$ logical qubits, composed of $n$ Pauli product rotations, one per column of the matrix, each consuming one magic state. 

Several works have investigated how to construct such matrices, including protocols producing multiple output states~\cite{Haah17b, Haah2018, Shi24, Sullivan_2024}. One parameter that has received particular attention is the \emph{depth} of the protocol, corresponding to the number of columns of the matrix, i.e.~to the number of input magic states consumed~\cite{NezamiHaah2021, baldelli2026constructingdecodingquantumtriorthogonal}. While the depth directly determines the number of time steps of distillation factories, the physical resource overhead is also greatly impacted by the number of \emph{support qubits} $N$ the rotations act on~\cite{Jacinto26}. Recent work has shown that \emph{qubit recycling} can substantially reduce the resource requirements of distillation protocols by measuring and re-initializing check qubits midway through the protocol and delaying initialization of the output qubits~\cite{Xu26, Jacinto26}.

In this paper, we tackle the space cost of distillation factories by introducing a family of binary matrices that correspond to distillation protocols of a single output state at arbitrary levels $r$ of the Clifford hierarchy~\cite{Cui17} and at arbitrary distance $d$. Our construction is recursive and based on a canonical procedure known as \emph{code doubling}, or \emph{level lifting}, which constructs higher-order orthogonal matrices from existing ones~\cite{Haah2018,Sullivan_2024}. We show that the support of this family of protocols can be efficiently reduced by qubit recycling, and that for a fixed Clifford level, the number of qubit needed to perform the protocols at arbitrary distance $d$ is constant. This can allow for substantially smaller resource overhead than previously known protocols of similar distance. Finally, we discuss possible extension of this framework to the case of protocols distilling more than one output state.

\section{$r$-orthogonal matrices}
Throughout this letter, we consider matrices over the binary field $\mathbb{F}_2$. Let $G$ be an $s\times n$ binary matrix. We say that $G$ is \emph{$r$-orthogonal}, for $r\geq 2$, if, for any $p\in[2,r]$ distinct rows of $G$, the number of columns in which all $p$ rows have support is a multiple of~$2^{r-p+1}$, i.e.
\begin{equation}
\begin{split}
\forall p\in[2,r],\quad & \forall i_1\neq\cdots\neq i_p\in[s], \\
& \qquad
\sum_{k=1}^{n} G_{i_1,k}\cdots G_{i_p,k} \equiv 0 \pmod{2^{r-p+1}}.
\end{split}
\label{eq:rortho}
\end{equation}
For $r=2$, this condition reduces to pairwise orthogonality of the rows, and one generally speak about self-orthogonal matrices. For $r=3$, one recovers the notion of triorthogonal matrices~\cite{BravyiHaah2012}. We also refer to arbitrary binary matrices as \emph{$1$-orthogonal}, since no constraint is imposed in this case. Up to a permutation of the rows, we write
\begin{equation}
G =\begin{bmatrix}
    G_1 \\ \hline G_0
\end{bmatrix}, 
\end{equation}
where the rows of $G_1$ and $G_0$ have odd and even Hamming weight respectively.

An $r$-orthogonal matrix naturally defines a magic-state distillation protocol producing $|\operatorname{rows}(G_1)|$ states. Each output state is
$\ket{(Z)^{\alpha_i/2^{(r-1)}}}$, defined as the $+1$ eigenstate of the operator $(Z)^{\alpha_i/2^{(r-1)}}$, where $\alpha_i$ is the Hamming weight of the $i$-th row of $G_1$. As derived in Ref.~\cite{Jacinto26}, given an $r$-orthogonal matrix $G$, the corresponding protocol acts on $s$ data qubits and consumes $n$ input magic states. For each column $i$, one applies the generalized rotation
\begin{equation}
R_{\pi/2^r}^{Z^{G_{\cdot,i}}} = \cos\left(\frac{\pi}{2^r}\right) I + \sin\left(\frac{\pi}{2^r}\right) \bigotimes_{j=1}^{s} Z_j^{G_{j,i}},
\label{eq:generalized_rotation}
\end{equation}
on the support specified by that column. Each such rotation can be implemented by injecting a single $\ket{(Z)^{1/2^{(r-1)}}}$ state. The even-weight rows, corresponding to $G_0$, are subsequently measured in the $Z$ basis and serve as \emph{check qubits}. Conditioned on the accepted measurement outcomes, the protocol implements a $(Z)^{\alpha/2^{(r-1)}}$ gate on each qubit associated with an odd-weight row of $G_1$, with $\alpha$ the Hamming weight of the row. Hence, the matrix $G$ defines a distillation procedure with
\begin{equation}
n(G)=n
\qquad\text{and}\qquad
k(G)=|\operatorname{rows}(G_1)|,
\end{equation}
where $n(G)$ is the number of input states (or \emph{depth}) and $k(G)$ is the number of output states. 

A distillation protocols is characterized by that fact that it suppresses error as $p^d$ for input noisy magic with error rate $p$ and for a distance $d$. Note that the distillation protocol presented above only account for $Z$ type errors. However, Ref.~\cite{Jacinto26} showed that when the underlying QECC code is a stabilizer code, it is sufficient to consider $Z$-type errors only for distillation of such states, since an $X$-type error on $Z$ diagonal states is mapped probabilistically to either no error or a $Z$-type error when measuring a $Z$ type stabilizer on the underlying QECC code. Secondly, the error model corresponding the implementation of $R_{\pi/2^r}^{Z^{G_{\cdot,i}}}$ gates thought state injection is a $Z$ error on each of its supporting qubit~\cite{Litinski2019}. We thus define $d(G)$ as the minimum number of faulty generalized rotations required to produce a $Z$ error on at least one output qubit while remaining undetected by all check-qubit measurements. This quantity can thus be expressed directly in terms of $G$: $d(G)$ is the minimum number of columns whose binary sum has zero support on the rows of $G_0$ and nonzero support on at least one row of $G_1$. The protocol therefore maps a matrix $G$ to a distillation procedure consuming $n(G)$ states $\ket{(Z)^{1/2^{(r-1)}}}$ with error $p$ to produce $k(G)$ states $\ket{(Z)^{\alpha/2^{(r-1)}}}$ with error scaling as $p^{d(G)}$ to leading order~\footnote{The prefactor of the $p^{d(G)}$ term can be computed using the MacWilliams identity, as mentioned in~\cite{BravyiHaah2012}}. We further denote by $s(G)=s$ the number of rows of $G$, which determines the number of \emph{support qubits} required to implement the protocol. In particular, the qubit volume overhead~\cite{Litinski2019} is given by
\begin{equation}
\mathcal{O}(G)=s(G)\cdot n(G).
\end{equation}
We thus associate to every $r$-orthogonal matrix $G$ an
$[n(G),k(G),d(G),s(G)]$ distillation procedure of
$\ket{(Z)^{\alpha/2^{(r-1)}}}$ states. 

Of particular interest are self-orthogonal matrices ($r~=~2$) enabling distillation of $\ket{S}=\ket{Z^{1/2}}$ states, which cannot be implemented efficiently on error-correcting codes tailored to biased-noise architecture~\cite{Barizien26}, such as ones based on cat qubits~\cite{Gouzien23}. Furthermore, triorthogonal matrices ($r=3$) lead to the distillation of the $\ket{T}:=\ket{Z^{1/4}}$ magic state~\cite{BravyiHaah2012}, key to most two-dimensional architectures~\cite{Fowler19,Bravyi24}. In this case, it is natural to assume that Clifford operations are available at negligible cost at the logical level, allowing for additional Clifford corrections within the distillation protocol~\cite{Litinski2019, NezamiHaah2021}. 

More generally, when distilling $\ket{(Z)^{\alpha/2^{(r-1)}}}$ states for odd $\alpha$, which belong to the $r$-th level of the Clifford hierarchy~\cite{GottesmanChuang1999, Cui17}, one may assume that operations from lower levels of the hierarchy are freely available, either because they can be implemented directly at the logical level or because the corresponding magic states can be distilled at lower cost. Equivalently, one may assume access to high-quality $(Z)^{1/2^{(r-2)}}$ gates. This additional freedom allows the orthogonality constraints in Eq.~\eqref{eq:rortho} to be relaxed, providing additional corrections are later performed. We therefore introduce \emph{weakly $r$-orthogonal} matrices, for which the overlaps in Eq.~\eqref{eq:rortho} reduce to modulo $2$ instead:
\begin{equation}
\begin{split}
\forall p\in[2,r],\quad
& \forall i_1\neq\cdots\neq i_p\in[s],\\
& \qquad
\sum_{k=1}^{n}
G_{i_1,k}\cdots G_{i_p,k}
\equiv 0 \pmod{2}.
\end{split}
\label{eq:weakrortho}
\end{equation}
In this framework, the states that are distilled can all be mapped, up to lower level corrections, to the canonical state $\ket{(Z)^{1/2^{(r-1)}}}$. We note that the terminology \emph{weakly orthogonal} is not universal; in particular, related works use the term to refer specifically to orthogonality conditions that allow for Clifford corrections only~\cite{Jacinto26}.

\section{Code doubling}
A general question regarding $r$-orthogonal matrices is how to construct such matrices at arbitrary distance. Previous work have shown canonical families exists, such as ones based on shortened Reed-Muller codes~\cite{Landahl13}, distilling transversal gates high in the Clifford hierarchy at any $d$. In the following, we tackle this question for matrices with a single output, i.e~$k(G)=1$. We present a general construction, to obtain higher distances codes based on existing ones. This construction is based on \emph{code doubling}, as introduced in Refs.~\cite{BravyiCross2015,Haah2018}, and recalled in the Theorem below.

\begin{restatable}{Theorem}{THMOne} \label{thm:code_doubling}
    Consider a weakly $r$-orthogonal matrix $G = \begin{bmatrix}
    G_1 \\ \hline G_0
\end{bmatrix}$ with a single output ($k(G)=1$) of distance $d$ and a weakly ($r-1$)-orthogonal matrix $H=\begin{bmatrix}
    \mathbf{1} \\ \hline H_0
\end{bmatrix}$ of odd depth $n(H)$ and distance $d+2$. Then the following matrix is weakly $r$-orthogonal 
\begin{equation}
    G' = \begin{bmatrix}
        G_1 & H_1 & H_1 \\ \hline
        G_1 & \mathbf{0} & \mathbf{1} \\
        G_0 & 0 & 0 \\ 
        0 & H_0 & H_0
    \end{bmatrix}
\end{equation}
and has distance $d+2$, where $0$ and $1$ are all respectively all zeros and all ones matrices of appropriate size, and the bold font is used for single row matrices. 
\end{restatable}
\noindent The proof of this result, which can be found in Appendix~\ref{app:thm1}, relies on the orthogonality properties of $G$ and $H$, together with the particular structure of the construction. In particular, the two copies of $G_1$ ensure that the first two rows retains the required $r$-orthogonality conditions, while the $(r-1)$-orthogonality of $H$ guarantees that $r$-orthogonality is satisfied with the additional $\bm 1$ row. The distance of the constructed matrix is increased to $d+2$ as any error pattern of $G$ would be detected by the additional $G_1$ check qubit unless two additional faulty columns happen later.

While the distance is increased by two by the code doubling procedures, it comes at the cost of increasing the support and depth of the matrix. Indeed, for the new matrix G', these quantities read:
\begin{equation} \label{eq:recursive_param}
\begin{split}
    & s(G') = s(G)+1+s(H_0), \\ 
    & n(G') = n(G)+2n(H).
\end{split}
\end{equation}

\section{Recursive family construction}
We now use the code-doubling construction to recursively generate families of $r$-orthogonal matrices with increasing distance. We first introduce a set of elementary matrices that serves as the base for the recursion. More precisely, we introduce the following $1$-orthogonal matrices:
\begin{equation}
    G^{(1,d)} = \begin{bmatrix}
        \mathbf{1} \\ \hline
        G_0^{(1,d)}
    \end{bmatrix}
\end{equation}
with $G_0^{(1,d)}$ a $(d-1)\times d$ binary matrix such that $G_0^{(1,d)}[i,i]=G_0^{(1,d)}[i,i+1]=1$ for $i\in\{1,\dots,d-1\}$ and all other terms are $0$. For odd distances $d=2m+1$, the depth of the matrix is odd, and the first line sums to $1$ while all other line have exactly two terms $1$ and sum to zero. The only way to get an error on the output qubit is to have an error at each column and thus the distance of such a code is indeed $d$. 

Starting from these matrices, we recursively construct a two-parameter family of $r$-orthogonal matrices. For $r\geq1$ and odd $d$, we define
\begin{equation} \label{eq:can_family_k1}
    G^{(r+1,d+2)} = \begin{bmatrix}
        G^{(r+1,d)}_1  & G^{(r,d+2)}_1 & G^{(r,d+2)}_1 \\ \hline
        G^{(r+1,d)}_1  & \mathbf{0} & \mathbf{1} \\
        G^{(r+1,d)}_0 & 0 & 0 \\ 
        0 & G^{(r,d+2)}_0 & G^{(r,d+2)}_0
    \end{bmatrix},
\end{equation}
where we initialize the recursive process by taking $G^{(r,1)}:=[1]$. Since we apply the code-doubling construction at each step, Theorem~\ref{thm:code_doubling} guarantees that the matrix $G^{(r,d)}$ is $r$-orthogonal, of distance $d$, with a single output. Note that at each step it indeed holds that $G^{(r,d)}_1 = \mathbf{1}$.

The recursive construction provides an explicit family of increasing-distance distillation protocols. Its resource requirements are determined by the number of columns and rows at each recursion step. In particular, the depth and support satisfy recursion relations obtained directly from Eq.~\eqref{eq:recursive_param}, which can be solved to give closed-form expressions for fixed $r$. In general, the parameters $s(G^{(r,d)}) $ and $n(G^{(r,d)})$ of the associated protocol scale as $d^r$.

Some interesting codes appear when looking at small parameter values of $r$ and $d$. For $r=2$, the construction gives the following protocols for $\ket{S}:=\ket{\sqrt{Z}}$ distillation: a $[7,1,3,4]$ protocol, corresponding to the Steane code~\cite{Steane96}; a $[17,1,5,9]$ protocol, corresponding to a 2D self-dual color code~\cite{Bombin_2006}; a $[31,1,7,16]$ protocol, corresponding to a quantum color code~\cite{Bombin_2006}, or equivalently to the generator matrix of some shortened classical Reed-Muller code~\cite{Landahl13}. 
For $r=3$, we obtain the following protocols for $\ket{T}:=\ket{\sqrt{\sqrt{Z}}}$ distillation: a $[15,1,3,5]$ protocol, corresponding to the smallest quantum Reed-Muller code~\cite{Bravyi05, Litinski2019}; a $[49,1,5,14]$ protocol, already reported in~\cite{BravyiHaah2012}. To our knowledge, the higher distances protocols, such as the $[111,1,7,30]$ protocol, were not reported before. These are compared to the results of other small codes known in the literature in~\Cref{tab:table1}.

More generally, in the case of triorthogonal matrices ($r=3$), the codes parameters scale as $n(G^{(3,d)}) = (d^3+6d^2+5d-6)/6$ and $s(G^{(3,d)}) = (d+1)(d+2)(d+3)/24$. Note that they may be suboptimal, as for instance, in terms of depth, one could use shorter $2$-orthogonal codes in the code doubling construction. For example, for $d=5$, there exists a $2$-orthogonal matrix, corresponding to the Golay code~\cite{Steane96}, which only has depth $23$. This could be use to produce a new branch of tri-orthogonal codes based on the doubling construction~\cite{Sullivan_2024}. This way, a $\ket{T}$ protocol of depth $95$ can be built. However, both the depth and the support of the matrix need to be considered for the overhead of such protocols.

\begin{table*}[] 
    \centering 
    \renewcommand{\arraystretch}{1.25} \setlength{\tabcolsep}{10pt} 
    \begin{tabular}{@{}ccccccc@{}} \toprule  \textbf{Matrix} & \textbf{Ref.} & \textbf{Support} $s$ & \textbf{Recycled} $s^\star$ & \textbf{Depth} $n$ & \textbf{Effective overhead} $s^\star \cdot n$ & \textbf{Distance} $d$ \\ 
    \midrule \midrule
    \multicolumn{6}{c}{$r=2$} \\
    \midrule $G^{(2,3)}$ & \cite{Steane96} & 4 & 3 & 7 & 21 & 3 \\$G^{(2,5)}$ & \cite{Bombin_2006} & 9 & 4 & 17 & 68 & 5 \\ $G^{(2,7)}$ & \cite{Bombin_2006} & 16 & \textbf{$\leq$ 4} & 31 & \textbf{$\leq$ 124} & 7 \\
    -- & \cite{Steane96, Jain_2025} & 12 & $\leq^{\star}$ 11 & 23 & $\leq$ 253 & 7\\
    $G^{(2,9)}$ & This work & \textbf{25} & \textbf{$\leq$ 4} & \textbf{49} & \textbf{$\leq$ 196} & \textbf{9} \\
    -- & \cite{Jain_2025} & 23 & $\leq$ 23 & 49 & $\leq$ 1127 & 9 \\
    \midrule 
    \multicolumn{6}{c}{$r=3$} \\
    \midrule $G^{(3,3)}$ & \cite{BravyiHaah2012} & 5 & 4 & 15 & 60 & 3 \\ $G^{(3,5)}$ & \cite{BravyiHaah2012} & 14 & 5 & 49 & 245 & 5 \\ $G^{(3,7)}$ & This work & \textbf{30} & \textbf{$\leq$ 6} & \textbf{111} & \textbf{$\leq$ 666} & \textbf{7} \\
    -- & \cite{Sullivan_2024} & 26 & $\leq^{\star}$ 11 & 95 & $\leq$ 1045 & 7\\
    $G^{(3,9)}$ & This work & \textbf{55} & \textbf{$\leq$ 6} & \textbf{209} & \textbf{$\leq$ 1254} & \textbf{9} \\
    -- & \cite{Jain_2025} & 49 & $\leq$ 23 & 185 & $\leq$ 4255 & 9 \\
    \bottomrule 
    \end{tabular}
    \caption{Parameters of the 1-output distillation protocols from the recursive family $G^{(r,d)}$ for small values of $r$ and $d$, and comparison with other small protocols known in the literature. For the recycled qubit value $s^ \star$, an exact value means that is was certified optimal using the numerical method presented in Ref.~\cite{Jacinto26} ; upper bounds are obtained either because of the code doubling construction of due to numerical exploration (indicated by a $^\star$), using the search algorithm of Ref.~\cite{Jacinto26}.}
    \label{tab:table1}
\end{table*}

\section{Qubit recycling}
It may happen that not all qubits are simultaneously in-use during the whole duration of distillation protocol. Such property can be leverage to reduced the effective qubit support of a given distillation protocol. This idea was recently formalized as \emph{qubit recycling}~\cite{Xu26} noticing that some of the check qubits can be initialized, used to performed some of the rotations, measured and initialized again. Furthermore, the output qubits can sometimes be initialized latter in the sequence. 

However, exploring qubit recycling for a given distillation protocol is not straightforward as many equivalent matrix generating that protocol exist, obtained by operations preserving \emph{r}-orthogonality and distance, namely adding even weighted rows to other rows and column permutation. Each equivalent protocol might admit different qubit support after recycling. In general, computing the minimal recycled qubit support for the equivalence class of a protocol, denoted in the following $s^\star(G)$, is hard, thought recent works have try to tackle this numerically~\cite{Jacinto26}.  

Our recursive construction based on code-doubling however allows to easily compute a qubit recycling reduction at every step of the recursive process. Indeed, one can perform the recursive construction using recycling as following: 
\begin{enumerate}
    \item in Eq.~\eqref{eq:can_family_k1}, adding the second row to the first preserve orthogonality and allows to initialize the output qubits after the first column block,
    \item the check qubits corresponding to the $G_0^{(r+1,d)}$ rows are measured and re-initialized,
    \item the output qubit is initialized and the last two column blocks are performed and measured 
\end{enumerate}
All in all, the recycled procedure can be represented as
\begin{equation} 
    G^{(r+1,d+2)} = \begin{bmatrix}
          & & \mathbf{1} & \mathbf{0} \\ \hline
        \mathbf{1} & &  \mathbf{0} & \mathbf{1} \\
        G^{(r+1,d)}_0 & || &  &  \\ 
        &  & G^{(r,d+2)}_0 & G^{(r,d+2)}_0
    \end{bmatrix},
\end{equation}
where we used the fact that $G^{(r,d)}_1 = \mathbf{1}$, the blank means that the qubits are not used at this time step and the vertical double bar means that the check qubits are measured midway through. Rearranging the lines and filling the blanks, one can alternatively write it as
\begin{equation}
\left[\begin{array}{c||ccc}
        \multirow{2}{*}{$\ \, G^{(r+1,d)}_0 \ $} 
    & & \mathbf{1}
    & \mathbf{0} \\ \cline{3-4} 
    & &  G^{(r,d+2)}_0 & G^{(r,d+2)}_0\\
        \multicolumn{1}{c}{\mathbf{1}} & & 
    \mathbf{0}
    & \mathbf{1}
\end{array}\right]
\end{equation}
With this representation, it becomes clear that: (i) the first column block above can be performed on the same recycled support as $G^{(r+1,d)}$; (ii) the last two column blocks require at most two more support qubit than the recycled support of $G_0^{(r,d+2)}$, which is itself upper bounded by the recycled support of $G^{(r,d+2)}$. 

Now, notice that the trivial code canonical family ($r~=~1$) can always be realized on $3$ qubits as only $2$ of the check qubits need to be active at the same time. Thus, one can recursively show that the protocol corresponding to $G^{(r,d)}$ can be performed on a recycled support 
\begin{equation}
    s^\star(G^{(r,d)}) \leq 2r,
\end{equation}
for all $r\geq 2$. As such, for a fixed level $r$ in the Clifford hierarchy, it is possible to distill states up to $r-1$ Clifford corrections, with arbitrary large distance $d$, on a fixed number of qubits. In the case of $r=3$, our construction allows to only use $6$ qubits to perform $\ket{T}$ state distillation protocols at to arbitrary distance $d$. Note that this construction is only an upper bound on the recycled support. For instance, the $15\to 1$ protocols ($d=3$) can be performed on only $4$ support qubits, while $s=5$ seems to be numerically optimal for the $49\to 1$ protocol ($d=5$)~\cite{Jacinto26}.

This key property makes out codes relevant in terms of effective qubit overhead $s^\star(G)\cdot n(G)$. For instance, for $r=2$, the self-orthognal code with Golay code construction for $d=7$ seems to be incompressible numerically on less than $s^\star=11$ for depth $n=23$, whereas ours has effective support $s^\star=4$ and depth $n=31$. Therefore our code has much smaller effective overhead $4\cdot 31 = 124 < 253 = 11\cdot 23$. The same happen for code that would be constructed by doubling the Golay code. For instance, for $r=3$, the $d=7$ code obtained in Ref.~\cite{Sullivan_2024} of depth $95$ is itself limited by the $11$ support qubit of the Golay code used to construct it. While our $d=7$ code has larger depth $n=111$, it's effective overhead $6\cdot111 = 666 < 1045= 11\cdot 95$ is again much smaller. \\

\section{Doubling in the $k>1$ case}
The restriction to a single output state, $k(G)=1$, is convenient for the recursive construction above, but is not optimal from the point of view of asymptotic distillation efficiency. Indeed, for a matrix with $k$ odd-weight rows, the output-to-overhead ratio is $k/(s\cdot n)$, so increasing $k$ can substantially improve the yield of a distillation protocol. This observation already appears in Ref.~\cite{BravyiHaah2012}, where a family of $d=2$ triorthogonal codes can yield a $\ket{T}$ distillation protocol with parameters $[3k+8,k,2,k+3]$. Asymptotically, the overhead-to-output ratio scale as $1/k$, and these codes therefore provide a simple example in the $d=2$ case where multiple output states lead to a parametrically better protocols than any fixed-$k$ construction. Subsequent works have manage to reduce the output-to-input ratio progressively, down to asymptotically optimal distillation protocols, corresponding to $\log(n/k)/\log(d) =0$~\cite{HaahHastings2018,Wills2025}, but little focus has been shown to the size of the qubit support for these protocols. 

This motivates extending the code-doubling construction to matrices with $k(G)>1$. The underlying algebraic mechanism is in fact not specific to the single-output case, as the doubling construction can be understood as specific case of level lifting, a map which combines an $r$-orthogonal matrix with an $(r+1)$-orthogonal matrix to produce an $(r+1)$-orthogonal matrix~\cite{Haah2018}. However, the key point here is to derive how the distance will behave under such mapping. Here, we show that is it in general possible to increase the distance from odd $d$ to even $d+1$, as stated below.
\begin{restatable}{Theorem}{THMTwo} \label{thm:code_doubling_2}
    Consider a weakly $r$-orthogonal matrix $G = \begin{bmatrix}
    G_1 \\ \hline G_0
\end{bmatrix}$ of odd distance $d$ and a weakly ($r-1$)-orthogonal matrix $H=\begin{bmatrix}
    H_1 \\ \hline H_0 
\end{bmatrix}$ with $k(H) = k(G)$, depth $n(H)$ of the same parity as 
$n(G)$ and distance at least $d+1$. Then the following matrix is weakly $r$-orthogonal 
\begin{equation}
    G' = \begin{bmatrix}
        G_1 & H_1 & H_1 \\ \hline
        \mathbf{1} & \mathbf{0} & \mathbf{1} \\
        G_0 & 0 & 0 \\ 
        0 & H_0 & H_0
    \end{bmatrix}
\end{equation}
and has distance at least $d+1$. The support of the new matrix is $s(G') = s(G) + s(H) - k(G) + 1$ while its depth is $n(G') = n(G) + 2n(H)$. The recycled qubit support of the new protocol can be upper bounded as $s^\star(G') \leq \max(s^\star(G)+1, s^\star(H)+1)$.
\end{restatable}

\noindent The proof of this result is detailed in Appendix~\ref{app:thm2}, and follows from similar steps to the one of \Cref{thm:code_doubling}. For the recycled support bound, note that the check qubits involved in the $G_0$ term can be measured and reset. The upper bound then comes from a direct analysis of the recycled support for the first column block and the last two column blocks.

This construction is particularly relevant to produce $d=2$ distillation protocols. In particular, starting from the trivial identity matrix $G := I_k$ of size $k\times k$, one can construct a weakly $r$-orthogonal code of distance $d=2$ by leveraging appropriate weakly $r-1$-orthogonal with $k$ outputs. For instance, for even $k$, one can use a recursive construction starting from the $1$-orthogonal family
\begin{equation}
    S^{(1,2,k)} = \begin{bmatrix}
        \multicolumn{3}{c}{I_k} \\ \hline 
        G_0^{(1,2)} & \cdots & G_0^{(1,2)}
    \end{bmatrix},
\end{equation}
and construct recursively weakly $r$-orthogonal matrices as 
\begin{equation}
    S^{(r+1,2,k)} = \begin{bmatrix}
        I_k & S_1^{(r,2,k)} & S_1^{(r,2,k)} \\
        \mathbf{1} & \mathbf{0}& \mathbf{1}\\ \hline
        0 &  S_0^{(r,2,k)} &  S_0^{(r,2,k)} 
    \end{bmatrix}.
\end{equation}
This family correspond to distillation protocols for $\ket{(Z)^{1/(r-1)}}$ states of parameters $[(2^r-1)k,k,2,k+r+1]$. For $r=2$ and $r=3$, this family contains the minimal depth matrices for $k=2$~\cite{NezamiHaah2021}, but is suboptimal afterwards.

Another construction we present here concerns protocols producing two outputs, i.e.~$k(G)=2$. 
\begin{restatable}{Theorem}{THMThree} \label{thm:code_doubling_3}
    Consider a weakly $r$-orthogonal matrix $G = \begin{bmatrix}
    G_1 \\ \hline G_0
\end{bmatrix}$ of even distance $d$ with $k(G)=2$. Denoting $(G_1)_i$ for $i=0,1$ the rows of $G_1$, we assume that $(G_1)_0 + (G_1)_1 = \mathbf{1}$. Then, consider any weakly ($r-1$)-orthogonal matrix $H=\begin{bmatrix} H_1 \\ \hline H_0 
\end{bmatrix}$ of odd depth $n(H)$ and distance at least $d+1$. It holds that the following matrix is weakly $r$-orthogonal 
\begin{equation}
    G' = \begin{bmatrix}
        \multirow{2}{*}{$G_1$} & H_1 & H_1 \\ &  \mathbf{0} & \mathbf{0}
        \\\hline
        (G_1)_0 & \mathbf{0} & \mathbf{1} \\
        G_0 & 0 & 0 \\ 
        0 & H_0 & H_0
    \end{bmatrix}, 
\end{equation}
and has distance at least $d+1$. The support of the new matrix is $s(G') = s(G) + s(H) - 1$ while its depth is $n(G') = n(G) + 2n(H)$. The recycled qubit support of the new protocol can be upper bounded as $s^\star(G') \leq \min[\max(s^\star(G), s^\star(H_0)+3), \max(s^\star(G)+1, s^\star(H)+1)]$. 
\end{restatable}

\noindent The proof of this result is given in Appendix~\ref{app:thm3} following similar logic as the other two theorems. The distance property directly follows from the hypothesis on the rows of $G_1$. Indeed, since the rows of $G_1$ are complementary, an even number $d$ of errors can only lead to: (i) an even number of errors on both outputs; (ii) an odd number of errors on both outputs. Case (i) amounts to no output error since $Z^2=I$, while case (ii) leads to an error on the additional check qubit $(G_1)_i$, requiring at least one further error to occur not to be detected. Finally, for the recycled support, adding the even row $(G_1)_0$ to the first row allows to initialize the first output qubit after the first column block. As previously, the check qubits involved in the $G_0$ term can be measured and reset, before the last two column blocks. The upper bound comes from doing the last two column blocks without initializing the second output qubit, measure and reset all qubit involved in $H_0$, and then do the first column block.

Using this, we propose the following construction for matrices of arbitrary distance $d$ in the case $k=2$. The recursion is initialized by the following matrices:
\begin{equation}
    P^{(1,d,2)} = \begin{bmatrix}
        I_2 \\ \hline 
        \multicolumn{1}{c}{G_0^{(1,d)}} 
    \end{bmatrix}, \quad \text{and  } P^{(r,1,2)} = \begin{bmatrix}
        I_2 
    \end{bmatrix},
\end{equation}
for $r\geq 1$. The complete family $P^{(r,d,2)}$ is built by recursion on $r$ and increasing $d$ by $1$ using Theorem~\ref{thm:code_doubling_2} with $H=P^{(r-1,d+1,2)}$ for odd $d$, and Theorem~\ref{thm:code_doubling_3} with $H=G^{(r-1,d+1)}$ for even $d$. This construction allows to obtain a family of distillation protocols producing $2$ output $\ket{(Z)^{1/2^{(r-1)}}}$ states, with depth and support scaling as $d^r$. 

Again, the key observation here is that the construction allows for efficient qubit recycling. The fact that, when going from odd to even distances, a $1$ increment in the recycled support can happen means that the recycled qubit support of these protocols can only be upper bounded by
\begin{equation}
    s^\star(P^{(r,d,2)}) \leq d^{r-1}. 
\end{equation}
However, for small parameter value, efficient reduction of the qubit count can still be achieved. A list of the codes obtained with this construction is given in \Cref{tab:k2case} for small parameters $r,d$. Note that for $r=3$, the obtained protocols have optimal depth, at least up to $d=3$~\cite{NezamiHaah2021}.

\begin{table}[h] 
    \centering 
    \renewcommand{\arraystretch}{1.25} \setlength{\tabcolsep}{10pt} 
    \begin{tabular}{@{}cccc@{}} \toprule \textbf{Support} $s$ & \textbf{Recycled} $s^\star$ & \textbf{Depth} $n$ & \textbf{Distance} $d$ \\ 
    \midrule \midrule
    \multicolumn{4}{c}{$r=2$} \\
    \midrule 4 & 4 & 6 & 2 \\ 7 & 4 & 12 & 3 \\ 11 & 5 & 20 & 4 \\ 16 & 5 & 30 & 5 \\
    \midrule 
    \multicolumn{4}{c}{$r=3$} \\
    \midrule
    5 & 5 & 14 & 2 \\ 9 & 5 & 28 & 3 \\ 19 & $\leq$ 6 & 68 & 4 \\ 28 &  $\leq$ 7 & 102 & 5 \\ 
    \midrule 
    \multicolumn{4}{c}{$r=4$} \\
    \midrule
    6 & 6 & 30 & 2 \\ 11 & $\leq$ 7 & 60 & 3 \\ 29 &  $\leq$ 8 & 196 & 4 \\ 
    \bottomrule 
    \end{tabular}
    \caption{Parameters of the 2-outputs distillation protocols from the recursive family $P^{(r,d,2)}$ for small values of $r$ and $d$. For the recycled qubit value $s^ \star$, an exact value means that is was certified optimal using the numerical method presented in Ref.~\cite{Jacinto26} ; upper bounds are obtained with the code doubling construction of Theorems~\ref{thm:code_doubling_2} and \ref{thm:code_doubling_3}.}
    \label{tab:k2case}
\end{table}

\section{Discussion and outlooks}
In this work, we show that $r$-orthogonal matrices can be constructed in ways that allow for efficient reduction of the qubit support using qubit recycling. As such, we provide a family of distillation protocol that produces one output state belonging to the $r$-th level of the Clifford hierarchy at arbitrary distance, and that can be effectively performed using a number of qubit that doesn't depend on $d$. Such protocols with very small active qubits appear particularly relevant when considering architectures with limited connectivity at the logical level. This was made possible by introducing a specific base layer of $1$-orthogonal matrix and leveraging a code doubling construction. We also discussed possible generalization to protocols having more than one output. Specifically, we introduced two new code doubling constructions and used them to family of codes, respectively of $k$ outputs and distance $2$; and $2$ outputs with arbitrary distance $d$. However, in the latter case, we were only able to show a linear gain on the qubit support when using qubit recycling. A rather natural question thus arise: is it possible to construct family of codes for arbitrary $d$ that have constant size recycled qubit support for any number of output $k$?

Note that recently, $d=2$ codes have been deeply explored in Ref.~\cite{Singh26}, even for states that are not diagonal, such as $\ket{CCZ}$ and $\ket{CS}$ states. It would be interesting to know which of these protocols can be obtained from similar constructions as the ones presented here, as they also rely on matrices with particular orthogonality properties.

Finally, our work shows that comparing different distillation protocols is a tricky task, as both the depth and the qubit support are involved in the resource overhead. We did not study here other efficient techniques to reduce the overhead, such as \emph{parallelization} of gates in the distillation sequence~\cite{Litinski2019}, or catalytic protocols, reusing some of the check qubits other distillation rounds~\cite{Singh26}. Future lines of work could include finding relevant construction tailored to reducing both depth and support using parallelization and recycling.

\begin{acknowledgements}
We thank Xavier Valcarce for the precious feedback on the manuscript. We acknowledge funding by the Swiss National Science Foundations (project 219366). Disclaimer: in opposite fashion of a recent single-author paper trend, the scientific results of this paper came from the brain of the author, looking at the sea for numerous hours, and not from some LLM. 
\end{acknowledgements}

\bibliographystyle{apsrev4-2}
\bibliography{biblio}{}


\newpage 

\appendix
\onecolumngrid
\newpage 

\begin{center}
    \Large \textbf{Supplemental Materials}
\end{center}

\vspace{30pt}

We start this supplemental section by laying out some definitions and easy-for-the-reader-to-check propositions that will come in handy later.

\begin{Definition}
    For any set of binary vectors $(g_i)\in \mathbb{F}_2^n$, we define the following star product:
    \begin{equation}
        \star(g_1,\dots, g_n) = \sum_{i=1}^n (g_1)_i \cdots (g_n)_i.
    \end{equation}
\end{Definition}

\begin{Proposition}
    We can rewrite the (weak) $r$-orthogonality of a matrix $G$ as respectively:
    \begin{equation}
        \forall p\in[2,r],\quad \forall i_1\neq\cdots\neq i_p\in[s], \quad \star(g_{i_1},\dots, g_{i_p}) \equiv 0 \mod{2^{r-p+1}},
    \end{equation}
    and 
    \begin{equation}
        \forall p\in[2,r],\quad \forall i_1\neq\cdots\neq i_p\in[s], \quad \star(g_{i_1},\dots, g_{i_p}) \equiv 0 \mod{2},
    \end{equation}
    where $g_1,\dots,g_n$ are the rows of $G$.
\end{Proposition}

\begin{Definition}
    For a binary vector $g\in \mathbb{F}_2^k$ over the binary field, we denote $\text{wt}(g) \equiv \sum_{i}^k g_i \mod{2}$ the Hamming weight of the vector $g$. 
\end{Definition}

\begin{Definition}
    For two binary vectors $g\in \mathbb{F}_2^k$ and $h\in \mathbb{F}_2^{k'}$ over the binary field, we denote $g \oplus h = (g_1,\dots,g_k, h_1, \dots, h_{k'})$ the contraction of $g$ and $h$. 
\end{Definition}

\begin{Proposition}
    For two binary vectors $g\in \mathbb{F}_2^k$ and $h\in \mathbb{F}_2^{k'}$ over the binary field, it holds that $\text{wt}(g \oplus h)\equiv \text{wt}(g) + \text{wt}(h) \mod{2}$
\end{Proposition}

\begin{Proposition} \label{prop:starplus}
    For any vectors $g,g',h,h'$ over the binary field it holds that $\star(g\oplus h, g'\oplus h') = \star(g, g') + \star(h, h')$.
\end{Proposition}

\begin{Proposition} \label{prop:starzerosone}
    For any vectors $g_1,\dots,g_n$ over the binary field it holds that $\star(1,g_1,\dots,g_n) = \star(g_1,\dots,g_n)$, $\star(0,\dots,0,g_1,\dots,g_n)=0$ and $\star(g_i,g_i)=\star(g_i)$.
\end{Proposition}

\noindent From now on, all operations will be modulo 2. For a matrix $G\in \mathbb{F}_2^{s \times n}$ over the binary field, we will identically use $G$ as a represent of the rows of $G$ when a given property holds for every possible choice of row. For instance $\text{wt}(G)=0$ means that every row of $G$ as even weight.

\section{Proof of Theorem~\ref{thm:code_doubling}} \label{app:thm1}

\THMOne*

\begin{proof}
    We prove the following: (i) $G'$'s first row has odd weight and all other have even weight ; (ii) the matrix $G'$ is weakly $r$-orthogonal ; (iii) the distance of the matrix $G'$ is $d+2$. \\
    \textbf{\underline{(i):}} Let's compute the weight of each row:
    \begin{equation}
        \text{wt}(G') = \text{wt} \begin{bmatrix}
        G_1 & \oplus & H_1 & \oplus & H_1 \\ \hline
        G_1 & \oplus & \mathbf{0} & \oplus & \mathbf{1} \\
        G_0 &\oplus & 0 &\oplus & 0 \\ 
        0 &\oplus & H_0 &\oplus & H_0
    \end{bmatrix} = \begin{bmatrix}
        \text{wt}(G_1) + 2 \text{wt}(H_1) \\ \hline
        \text{wt}(G_1) + n(H) \\
        \text{wt}(G_0) \\ 
        2\text{wt}(H_0) 
    \end{bmatrix} = \begin{bmatrix}
        1 \\ \hline
        0 \\ 0 \\ 0 
    \end{bmatrix},
    \end{equation}
    where we used that $n(H)$ is odd. \\
    \textbf{\underline{(ii):}} Let $p\in[2,r]$. To compute the star product of $p$ rows, we distinguish 4 cases:
    \begin{itemize}
        \item All vectors belong to third and forth block rows, and we write, with $p_3 \in [0,p]$:
        \begin{equation}
        \begin{split}
            \star(\underset{p_3}{\underbrace{(G_0 \oplus 0\oplus 0),\dots,(G_0\oplus 0\oplus 0)}},\underset{p-p_3}{\underbrace{(0\oplus H_0\oplus H_0, \dots,0\oplus H_0\oplus H_0)}} ) & = \star(\underset{p_3}{\underbrace{G_0 ,\dots,G_0}},\underset{p-p_3}{\underbrace{0,\dots,0}}) + 2\left(\star(0,\dots,0,H_0 ,\dots,H_0)\right),\\
            & = \delta_{p-p_3=0}\star(\underset{p}{\underbrace{G_0 ,\dots,G_0}}),
        \end{split}
        \end{equation}
        where we used \Cref{prop:starplus} and \ref{prop:starzerosone}, and that everything is mod 2. We are left with $\star(G_0 ,\dots,G_0)$ which is also $0$ since $p\in[2,r]$ and $G$ is weakly $r$-orthogonal by hypothesis.
        \item A vector of the first row is included, and we write with $p_3 \in [0,p-1]$:
        \begin{equation}
        \begin{split}
            \star(G_1 \oplus H_1 \oplus H_1, & \underset{p_3} {\underbrace{(G_0 \oplus 0\oplus 0),\dots,(G_0\oplus 0\oplus 0)}},\underset{p-1-p_3}{\underbrace{(0\oplus H_0\oplus H_0, \dots,0\oplus H_0\oplus H_0)}} ) \\
            & = \delta_{p-1-p_3=0}\star(G_1,\underset{p-1}{\underbrace{G_0 ,\dots,G_0}}) + 2\delta_{p_3=0}\left(\star(H_1, H_0 ,\dots,H_0)\right)=0,
        \end{split}
        \end{equation}
        using that $(p-1)+1=p \in [2,r]$ and $G$ is weakly $r$-orthogonal.
        \item A vector of the second row is included, and we write with $p_3 \in [0,p-1]$:
        \begin{equation}
        \begin{split}
            \star(G_1 \oplus \mathbf{0} \oplus \mathbf{1}, & \underset{p_3} {\underbrace{(G_0 \oplus 0\oplus 0),\dots,(G_0\oplus 0\oplus 0)}},\underset{p-1-p_3}{\underbrace{(0\oplus H_0\oplus H_0, \dots,0\oplus H_0\oplus H_0)}} ) \\
            & = \star(G_1,\underset{p_3}{\underbrace{G_0 ,\dots,G_0}},0,\dots,0) + \star(\mathbf{0},0,\dots,0,\underset{p-1-p_3}{\underbrace{H_0 ,\dots,H_0}}) + \star(\mathbf{1},0,\dots,0,\underset{p-1-p_3}{\underbrace{H_0 ,\dots,H_0}}) \\
            & = \delta_{p-1-p_3=0}\star(G_1,\underset{p-1}{\underbrace{G_0 ,\dots,G_0}}) + \delta_{p_3=0} \star(\underset{p-1}{\underbrace{H_0 ,\dots,H_0}}) = 0,
	\end{split}
        \end{equation}
	where we used \Cref{prop:starzerosone} and that $(p-1)+1=p \in [2,r]$, $p-1 \in[1,r-1]$, $G$ is weakly $r$-orthogonal, $H$ is weakly ($r-1$)-orthogonal, and $\star(H_0)=\text{wt}(H_0)=0$.
        \item One vector of both first and second rows are included, and we write with $p_3 \in [0,p-2]$:
	\begin{equation}
        \begin{split}
            \star(G_1 \oplus H_1 \oplus H_1,& G_1 \oplus \mathbf{0} \oplus \mathbf{1}, \underset{p_3} {\underbrace{(G_0 \oplus 0\oplus 0),\dots,(G_0\oplus 0\oplus 0)}},\underset{p-2-p_3}{\underbrace{(0\oplus H_0\oplus H_0, \dots,0\oplus H_0\oplus H_0)}} ) \\
            & = \star(G_1, G_1,\underset{p_3}{\underbrace{G_0 ,\dots,G_0}},0,\dots,0) + \star(H_1,\mathbf{0},0,\dots,0,\underset{p-2-p_3}{\underbrace{H_0 ,\dots,H_0}}) + \star(H_1,\mathbf{1},0,\dots,0,\underset{p-2-p_3}{\underbrace{H_0 ,\dots,H_0}}) \\
            & = \delta_{p-2-p_3=0}\star(G_1, G_1,\underset{p-2}{\underbrace{G_0 ,\dots,G_0}}) + \delta_{p_3=0} \star(H_1,\underset{p-2}{\underbrace{H_0 ,\dots,H_0}}) = 0, \\
	& = \delta_{p-2-p_3=0}\star(G_1,\underset{p-2}{\underbrace{G_0 ,\dots,G_0}}) + \delta_{p_3=0} \star(H_1,\underset{p-2}{\underbrace{H_0 ,\dots,H_0}}),
	\end{split}
        \end{equation}
	where we used \Cref{prop:starzerosone}. Now three case are to be distinguished as
        \begin{equation}
        = \left\{\begin{split}
           & \star(G_1,\underset{p-2}{\underbrace{G_0 ,\dots,G_0}}) = 0, \quad \text{if } p_3 = p-2 \text{ and } p \in [3,r-2], \\
	&  \star(H_1,\underset{p-2}{\underbrace{H_0 ,\dots,H_0}})= 0, \quad \text{if } p_3 = 0 \text{ and } p \in [3,r-2], \\
	& \star(G_1) +  \star(H_1) = 1+1 = 0, \quad \text{if } p_3 = 0 \text{ and } p=2, \\
	\end{split} \right.
        \end{equation}
where the first two cases utilize the orthogonality of $G$ and $H$ respectively, and the last line comes from the weight of odd rows.
    \end{itemize}
\textbf{\underline{(iii):}} For the distance, we first focus on the first column row. In order to have an undetected error on $G_0$ affecting the outputs, at least $d$ errors are needed. Since the output qubit gets an error, the first check qubit is triggered, unless one more error happen in the third column block. However, if such error happen, one check qubit from the $H_0$ block would trigger as $d(H) =d+2 >1$. One more error must thus happen, totaling at least $d+2$ errors. However, an error can occurs with $d+1$ errors on the first column block. Now, we assume no error happen on the first column block and look at the last two column block. Let's label $J_2$ and $J_3$ the set of columns in which an error occur in column block 2 and 3 respectively. If $j\in J_2 \cap J_3$, then it either produces no errors or a pair of error on the ouput, amounting to no error. As such, we assume that $J_2 \cap J_3=0$. This means that the error pattern takes distinct column of $H_0$ and as such can produce an output error iff at least $d(H)=d+2$ error occur.

Conversely, it holds true that on the second block column, an error pattern of $d+2$ error can lead to an output error as the distance of $H$ is $d+2$, hence the distance is exactly $d+2$.
\end{proof}

\section{Proof of \Cref{thm:code_doubling_2}} \label{app:thm2}

\THMTwo*

\begin{proof}
    The proof resembles highly what's done in the previous section. Namely, we prove the following: (i) $G'$'s first $k(G)$ rows have odd weight and all other have even weight ; (ii) the matrix $G'$ is weakly $r$-orthogonal ; (iii) the distance of the matrix $G'$ is at least $d+1$ ; (iv) the bound on the recycled support.\\
    \textbf{\underline{(i):}} Let's compute the weight of each row:
    \begin{equation}
        \text{wt}(G') = \text{wt} \begin{bmatrix}
        G_1 & \oplus & H_1 & \oplus & H_1 \\ \hline
        G_1 & \oplus & \mathbf{0} & \oplus & \mathbf{1} \\
        G_0 &\oplus & 0 &\oplus & 0 \\ 
        0 &\oplus & H_0 &\oplus & H_0
    \end{bmatrix} = \begin{bmatrix}
        \text{wt}(G_1) + 2 \text{wt}(H_1) \\ \hline
        n(G) + n(H) \\
        \text{wt}(G_0) \\ 
        2\text{wt}(H_0) 
    \end{bmatrix} = \begin{bmatrix}
        1 \\ \hline
        0 \\ 0 \\ 0 
    \end{bmatrix},
    \end{equation}
    where we used that $n(G)$ and $n(H)$ have the same parity.\\
    \textbf{\underline{(ii):}} This point mostly follows similar calculations has in the proof of Theorem~\ref{thm:code_doubling}. For all $p\in[2,r]$, one can distinguish cases, the only relevant case that differs is when $p_1 \geq 1$ vectors from output rows are involved. Then two subcases remain:
    \begin{itemize}
        \item No vector from the first check row is involved and with $p_3 \in [0,p-p_1]$:
        \begin{equation}
        \begin{split}
            \star(& \underset{p_1}{\underbrace{(G_1 \oplus H_1 \oplus H_1), \dots, (G_1 \oplus H_1 \oplus H_1)}} ,\underset{p_3} {\underbrace{(G_0 \oplus 0\oplus 0),\dots,(G_0\oplus 0\oplus 0)}},\underset{p-p_1-p_3}{\underbrace{(0\oplus H_0\oplus H_0, \dots,0\oplus H_0\oplus H_0)}} ) \\
            & = \delta_{p-p_1-p_3=0}\star(\underset{p_1}{\underbrace{G_1, \dots, G_1}}, \underset{p_3} {\underbrace{G_0 ,\dots,G_0}})+ 2 \delta_{p_3=0} \star(H_1, \dots,H_1, H_0, \dots,H_0) =0,
        \end{split}
        \end{equation}
        since by $r$-orthogonality of $G$ the star product of $p_1+p_3=p \in [2,r]$ of its rows is $0$.
        \item The vector from the first check row is involved and with $p_3 \in [0,p-p_1-1]$: 
        \begin{equation}
        \begin{split}
            \star(& \underset{p_1}{\underbrace{(G_1 \oplus H_1 \oplus H_1), \dots, (G_1 \oplus H_1 \oplus H_1)}},(\mathbf{1},\mathbf{0},\mathbf{1}),\underset{p_3} {\underbrace{(G_0 \oplus 0\oplus 0),\dots,(G_0\oplus 0\oplus 0)}},\underset{p-p_1-p_3-1}{\underbrace{(0\oplus H_0\oplus H_0, \dots,0\oplus H_0\oplus H_0)}} ) \\
            & = \delta_{p-p_1-p_3-1=0}\star(\underset{p_1}{\underbrace{G_1, \dots, G_1}}, \mathbf{1},\underset{p_3} {\underbrace{G_0 ,\dots,G_0}})+ \delta_{p_3=0} \star(\underset{p_1}{\underbrace{H_1, \dots, H_1}}, \mathbf{1},\underset{p-p_1-1} {\underbrace{H_0 ,\dots,H_0}})\\
            & = \delta_{p-p_1-p_3-1=0}\star(\underset{p_1}{\underbrace{G_1, \dots, G_1}}, \underset{p_3} {\underbrace{G_0 ,\dots,G_0}})+ \delta_{p_3=0} \star(\underset{p_1}{\underbrace{H_1, \dots, H_1}}, \underset{p-p_1-1} {\underbrace{H_0 ,\dots,H_0}}),
        \end{split}
        \end{equation}
        where we used~\Cref{prop:starzerosone}. Note that since both star produt involve $p-1$ terms, and $p-1 \in [1,r-1]$, the above sum is $0$ by weak $r-1$ orthogonality of $H$ and $G$, unless $p-1=1$, i.e.~$p=2$. If so, $p_3 = 0,1$ and the above term is 
        \begin{equation}
             \left\{\begin{split}
                 & \star(G_1) + \star(H_1) = \text{wt}(G_1) + \text{wt}(H_1)= 1+1= 0, \quad \text{if } p_3=0, \\
                 & \star(G_0) = \text{wt}(G_0) = 0, \quad \text{if } p_3=1, 
             \end{split}\right.
        \end{equation}
    \end{itemize}
    \textbf{\underline{(iii):}} For the distance, we first focus on the first column row. In order to have an undetected error on $G_0$ affecting the outputs, at least $d$ errors are needed. Since $d$ is odd, this means the check qubit is triggered, unless one more error happen, totaling at least $d+1$ errors. Now, we assume no error happen on the first column block and look at the last two column block. Let's label $J_2$ and $J_3$ the set of columns in which an error occur in column block 2 and 3 respectively. If $j\in J_2 \cap J_3$, then it either produces no errors or a pair of error on each ouput, amounting to no error. As such, we assume that $J_2 \cap J_3=0$. This means that the error pattern takes distinct column of $H_0$ and as such can produce an output error iff at least $d(H)\geq d+1$ error occur. \\
    \textbf{\underline{(iv):}} The recycled pattern giving the upper bound is given by:
    \begin{equation} 
    G' = \begin{bmatrix}
         G_1 & & H_1 & H_1 \\ \hline
        \mathbf{1} & &  \mathbf{0} & \mathbf{1} \\
        G_0 & || &  &  \\ 
        &  & H_0 & H_0
    \end{bmatrix},
\end{equation}
The bound is obtained by looking at the first column block / (2 and 3) column block required number of qubits.
\end{proof}

\section{Proof of \Cref{thm:code_doubling_3}} \label{app:thm3}

\THMThree*

\begin{proof}
    The proof resembles highly what's done in the previous two sections. Namely, we prove refer the reader there for (i) the two first rows of $G'$ odd weight and all other have even weight. It remains to prove: (ii) the matrix $G'$ is weakly $r$-orthogonal ; (iii) the distance of the matrix $G'$ is at least $d+1$ ; (iv) the bound on the recycled support.\\
    \textbf{\underline{(ii):}} Let $p\in[2,r]$. Considering what's be done in the other two proofs, we only derive the case where the second output and the additional first check qubit are involved. Namely, for $p_3 \in [0,p-2]$, we write:
    \begin{equation}
        \begin{split}
            \star((G_1)_1 \oplus \mathbf{0} \oplus \mathbf{0} ,& (G_1)_0 \oplus \mathbf{0} \oplus \mathbf{1}, \underset{p_3} {\underbrace{(G_0 \oplus 0\oplus 0),\dots,(G_0\oplus 0\oplus 0)}},\underset{p-2-p_3}{\underbrace{(0\oplus H_0\oplus H_0, \dots,0\oplus H_0\oplus H_0)}} ) \\
            & = \star((G_1)_1, (G_1)_0,\underset{p_3}{\underbrace{G_0 ,\dots,G_0}},\underset{p-2-p_3}{\underbrace{0,\dots,0}}) \\
            & = \delta_{p-2-p_3=0}\star((G_1)_1, (G_1)_0,\underset{p-2}{\underbrace{G_0 ,\dots,G_0}}) = 0,
	\end{split}
        \end{equation}
        as weak $r$-orthogonality of $G$ ensures that the star product of $p-2+2=p\in[2,r]$ terms in $0$.\\
    \textbf{\underline{(iii):}} For the distance, the only difference with the two other proof in on the first column row. Indeed, assume that an error pattern occurs on $G$, then it involves at least $d$ errors, and affect an error to one (or two) of the outputs. However, the key hypothesis $(G_1)_0 + (G_1)_1 = \mathbf{1}$ ensure that each error happening affect one and only one output qubit. Since $d$ is even, it means that they are both either affected by an even number of errors, amounting to no errors, or by an odd number of error. In the later case, the additional check qubit $(G_1)_0$ on the first check row would also be triggered, and the rest of the proof follows like before.\\
    \textbf{\underline{(iv):}} The bound on the recycled support is obtained by minimizing over two different patterns:
    \begin{itemize}
        \item On can permute the columns to start with the last two column block, without initializing the second output qubit:
        \begin{equation}
    G' = \begin{bmatrix}
        H_1 & H_1 & & \multirow{2}{*}{$G_1$} \\ &  &  &  
        \\\hline
         \mathbf{0} & \mathbf{1}& & (G_1)_0 \\
         &  & & G_0 \\ 
         H_0 & H_0 & || & 
    \end{bmatrix}. 
\end{equation}
        A (1+2) / (3) column blocks support minimization in the above allows to bound the recycled support as $s^\star(G^\star) \leq \max(s^\star(H)+1, s^\star(G)+1) $. 
        \item On can add the first check qubit row to the first output qubit to initialize it later: 
        \begin{equation}
    G' = \begin{bmatrix}
         & & H_1 & H_1+\mathbf{1} \\ (G_1)_1 &  & \mathbf{0} & \mathbf{0}
        \\\hline
        (G_1)_0 & & \mathbf{0} & \mathbf{1} \\
        G_0 & || &   &  \\ 
         & & H_0 & H_0
    \end{bmatrix}.
\end{equation}
A (1) / (2+3) column blocks support minimization in the above allows to bound the recycled support as $s^\star(G^\star) \leq \max(s^\star(G), s^\star(H_0)+3) $. 
    \end{itemize}
\end{proof}

\end{document}